\PassOptionsToPackage{dvipsnames}{xcolor}
\documentclass[onefignum,onetabnum,twoside,reprint]{siamart171218}

\usepackage[dvipsnames]{xcolor}

\usepackage[T1]{fontenc}
\usepackage[utf8]{inputenc}
\usepackage[american]{babel}

\usepackage{amsmath}
\usepackage{amssymb}
\usepackage{bbm}
\usepackage{booktabs}
\usepackage{here}
\usepackage{hyperref}
\usepackage[square,numbers,sort&compress]{natbib}
\usepackage{subfig}

\newcommand{\R}{\mathbb{R}}
\newcommand{\C}{\mathbb{C}}
\newcommand{\e}{\varepsilon}
\newcommand{\eff}{\varepsilon^\text{eff}}

\renewcommand{\vec}[1]{\boldsymbol{#1}}

\newcommand{\vn}{\vec \nu}
\newcommand{\vx}{\vec x}
\newcommand{\vy}{\vec y}
\newcommand{\vz}{\vec z}

\newcommand{\dx} {\,{\mathrm d}x}
\newcommand{\dox}{\,{\mathrm d}o_x}

\newcommand{\doy}{\,{\mathrm d}o_y}
\newcommand{\Ds}{\,\Delta_T}

\newcommand{\cH}{\mathcal H}
\newcommand{\Gper}{G_{\text{per}}}

\newcommand{\Si}{{\Sigma}}
\newcommand{\Silarge}{{\Sigma_\ast}}
\newcommand{\js}[1]{\left[#1\right]_{\Si}}

\newcommand{\im}{\textrm i}

\newcommand{\bS}{\bold S}

\usepackage[page]{appendix}

\headers
  {Lorentz Resonance in the Homogenization of Plasmonic Crystals}
  {Li, Lipton, Maier}

\title{Lorentz Resonance in the Homogenization of Plasmonic Crystals}

\author{
  Wei Li%
  \thanks{%
    Department of Mathematical Sciences, DePaul University, Chicago, IL
    60614, USA}
  \and Robert Lipton%
  \thanks{%
    Department of Mathematics, Louisiana State University, Baton Rouge, LA
    70803, USA}
  \and Matthias Maier%
  \thanks{%
    Department of Mathematics, Texas A\&M University, College %
    Station, TX 77843, USA.}
}

\begin{document}

\maketitle

\begin{abstract}
  We explain the Lorentz resonances in plasmonic crystals that
  consist of 2D nano dielectric inclusions as the interaction between
  resonant material properties and geometric resonances of electrostatic
  nature. One example of such plasmonic crystals are graphene nanosheets
  that are periodically arranged within a non-magnetic bulk dielectric. We
  identify local geometric resonances on the length scale of the small
  scale period. From a materials perspective, the graphene surface exhibits
  a dispersive surface conductance captured by the Drude model. Together
  these phenomena conspire to generate Lorentz resonances at frequencies
  controlled by the surface geometry and the surface conductance.

  The Lorentz resonances found in the frequency response of the effective
  dielectric tensor of the bulk metamaterial is shown to be given by an
  explicit formula, in which material properties and geometric resonances
  are decoupled. This formula is rigorous and obtained directly from
  corrector fields describing local electrostatic fields inside the
  heterogeneous structure.

  Our analytical findings can serve as an efficient computational tool to
  describe the general frequency dependence of periodic optical devices. As
  a concrete example, we investigate two prototypical geometries composed
  of nanotubes and nanoribbons.
\end{abstract}


\section{Introduction}\label{sec:intro}

Novel frequency dependent electromagnetic behavior can be generated by
patterned dispersive dielectric metamaterials undergoing localized
geometric resonance. Here the period of the pattern lies below the wavelength of operation. 
Examples include plasmonic
metasurfaces~\cite{zhao2011,zhao2014}, band gaps generated by periodic
configurations of local plasmon resonators \cite{dregley2012}, and beam
steering \cite{pors2013}. In this work we contribute to the \emph{analytic}
understanding of such periodic optical devices by investigating the role of
local (frequency independent) geometric features and (frequency dependent)
material properties. In particular, we explain the appearance of Lorentz
resonances generated by periodically patterned dispersive dielectrics
\emph{as the interaction} between resonant material properties and local
geometric resonances of electrostatic nature.

Concretely, we shall examine the optical frequency response of plasmonic
crystals formed by 2D material inclusions (such as graphene) embedded in a
non-magnetic bulk dielectric host. We use a Drude model for the local
conductivity response of the 2D material but allow for a fairly general
periodic geometry including, for example, graphene nanoribbons, or graphene
nanotubes. In such geometries, frequency independent geometric resonances
will be identified and characterized that occur on the length scale of the
period of the 2D material inclusions. These local resonances are novel as
they exist both on the surface of the sheets and in the bulk. Together with
the dispersive surface conductance of the 2D material, both phenomena
conspire to generate Lorentz resonances in the effective optical frequency
response of the metamaterial. The resonance frequencies are controlled by
the surface geometry and the surface conductance.

The Lorentz resonances for the effective dielectric tensor or equivalently
the effective index of refraction for the bulk metamaterial are shown to be
given by an explicit formula. This formula is rigorous and obtained
directly from the corrector fields describing local electrostatic fields inside
the heterogeneous structure. The local boundary value problem for the
correctors follow from the periodic homogenization theory for Maxwell's
equations developed
in~\cite{amirat2017,maierxxa,wellander2001,wellander2002,wellander2003}.
The formula for the effective dielectric constant obtained here is notable
in that the local geometric resonances and local surface conductivity are
uncoupled. This offers the opportunity for efficient computation of the
effective dielectric constant through the computation of the local
geometric resonances that are independent of the  specific material
properties. The interaction between geometry and material dispersion is
displayed explicitly in the rigorously derived formula.

In detail, our contributions with the current work can be summarized as
follows:
\vspace{-1em}
\begin{itemize}
  \item
    We describe the interplay between frequency-independent geometric
    nanoscale resonances and frequency-dependent local conductivity models
    that results in Lorentz resonances in the effective optical frequency
    response. We derive an explicit formula for the frequency response
    rigorously from a mathematical homogenization theory for Maxwell's
    equations for periodic 2D material inclusions.
  \item
    The spectral decomposition is enabled by identifying an underlying
    compact self-adjoint operator on a proper function space. This was done
    by symmetrizing a non-hermitian operator.
  \item
    We discuss how to use the analytic result for computing approximations
    on the frequency response of periodic optical configurations. This
    approach offers a significant saving in computational resources because
    only one frequency-independent geometric eigenvalue problem has to be
    computed, in contrast to computing the corrector field for a huge
    number of fixed frequencies~\cite{maierxxa,maier19c}.
  \item
    We examine two prototypical geometries---a nanotube, and a nanoribbon
    configuration---in more detail. The latter one is analytically and
    computationally much more challenging due to singularities at interior
    2D material edges. We discuss decay estimates and examine the
    approximation quality of our computational approach.
\end{itemize}

\subsection{Background: Homogenization of plasmonic crystals}

The following analytical investigation is based on a rigorous periodic
homogenization theory~\cite{amirat2017, maierxxa, wellander2001,
wellander2002, wellander2003}. For the sake of simplicity, we will base our
analytical investigation on a slightly simplified setting that we quickly
outline here.

\begin{figure}[t]
  \centering
    \subfloat[]{
      \includegraphics{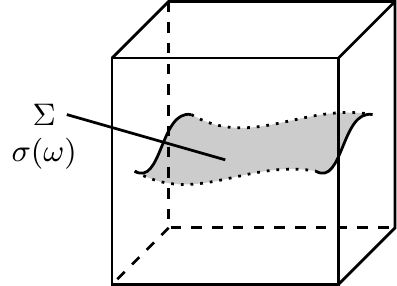}}
  \hspace{2em}
  \subfloat[]{
      \includegraphics{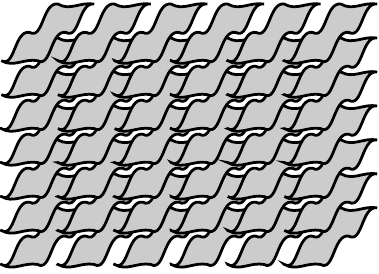}}
  \hspace{2em}
  \subfloat[]{%
    \includegraphics{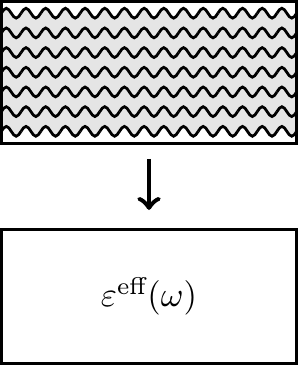}}
  \caption{The homogenization procedure: (a) the nanoscale unit cell $Y$
    consisting of 2D metallic inclusions $\Sigma$ with surface conductivity
    $\sigma(\omega)$ in an ambient host material with permittivity $\e$;
    (b) the plasmonic crystal formed by many scaled and repeated copies of
    $Y$ in every space dimension; (c) a schematic of the homogenization
    process in which the nanoscale structure is replace by a homogeneous
    material with effective permittivity $\eff$.}
  \label{fig:unit_cell}
\end{figure}
Consider a three-dimensional plasmonic crystal consisting of periodic
copies of a \emph{representative volume element} $Y$, which incorporates
nanoscale inclusions given by 2D material surfaces (see
Figure~\ref{fig:unit_cell}) of reasonably arbitrary shape (specified in
Sec. \ref{sec:intro}\ref{sec:summary} and Appendix \ref{app:open}). The
conductivity of the surfaces is assumed to obey the Drude model:
\begin{align*}
  \sigma(\omega)\;=\;
  \frac{\im\,\omega_p}{\omega+\im/\tau},
\end{align*}
where $\im$ denotes the imaginary unit, $\omega$ is the angular frequency,
$\omega_p = 4\,\alpha \approx 4/137$ is a (rescaled) Drude weight, and
$\tau$ is a material-dependent relaxation time. Here, we have
non-dimensionalized all quantities by applying a convenient
rescaling~\cite{maierxxc}: $\tilde\omega = \frac{\hbar\omega}{E_F}$, where
$E_{\text{F}}$ denotes the Fermi energy associated with the 2D material and
$\hbar$ is the reduced Planck constant; $\tilde\sigma(\tilde\omega) =
\sqrt{\frac{\mu_0}{\e_0}}\,\sigma(\omega)$, where $\mu_0$ and $\e_0$ denote
the vacuum permeability and permittivity, respectively. We set the length,
height, and width of the representative volume element to one, $Y=[0,1]^3$.
Furthermore, we assume that the dielectric host has a uniform and isotropic
relative permittivity $\varepsilon$.

It can then be shown \cite{maierxxa,maierxxc} that for sufficiently small
representative volume element $Y$ and sufficiently many repetitions of $Y$,
i.\,e., a sufficiently large plasmonic crystal, the \emph{effective}
conductivity of the plasmonic crystal is given by a uniform,
frequency-dependent conductivity tensor
\begin{align}
  \label{eq:intro:effective_perm}
  \eff_{ij}(\omega)\;=\;
  \varepsilon\,\delta_{ij}\;-\;
  \frac{\sigma(\omega)}{\im\omega}\,\int_\Si \big\{P_T(\vec e_j)\,
  +\,\nabla_T{\chi_j(\omega,\vx)}\big\}\cdot P_T(\vec e_i)\dox,
  \quad i,j=1,2,3.
\end{align}
Here,  {$\vx$ represents the spatial coordinates,} $\delta_{ij}$
 {is} Kronecker's Delta, $\vec e_j$ is the $j$-th unit
vector, $\Si$ denotes the 2D material surface (embedded in $Y$), $P_T$ is
the projection of a vector onto the two-dimensional tangential space of
$\Si$, and $\nabla_T=P_T\nabla$ denotes the tangential gradient (with
respect to $\Si$).

The $Y$-periodic corrector field $\vec\chi(\vec x)$ for closed $\Sigma$
is the solution of the \emph{cell problem} \cite{maierxxa},
\begin{align}
  \label{eq:intro:strong_full}
  \begin{cases}
    \begin{aligned}
    \Delta \chi_j(\vx)\big)
      &= 0
      && \text{in }Y\setminus\Sigma~,
      \\[0.5em]
      \js{\chi_j(\vx))} 
      &= 0
      && \text{on }\Sigma~,
      \\[0.5em]
      \e\js{\vn\cdot\nabla\chi_j(\vx)}
      &=
      \frac{\sigma}{i\omega}\nabla_T\cdot\big(P_T\vec e_j+ \nabla_T\chi_j(\vx)\big)
      && \text{on }\Sigma~,\color{black}
    \end{aligned}
  \end{cases}\hspace{-2.75em}
\end{align}
 {where}, $\vn$ is the unit outward normal of $\Si$ at $\vx$,
and $\left[f\right](\vx)$ denotes the jump of a quantity $f$ across the
surface $\Si$ along the normal direction of $\Si$, viz.,
\begin{align*}
  \left[f\right](\vx)
  \,:=\,
  \lim_{\alpha\searrow0}
  \Big(f(\vx+\alpha\vn) - f(\vx-\alpha\vn)\Big)\quad
  \text{for }
  \vx\in\Si.
\end{align*}

The novelty of plasmonic materials is that they are used to control light
at wavelengths much larger than the characteristic length scale of the
period. Thus understanding wave dispersion for such systems through
frequency-dependent effective behavior is quite natural. Recently frequency
dependent effective dispersive behavior of a finite number of metallic
sheets embedded in a dielectric host is compared to direct numerical
simulation and shown to agree up to a negligible error \cite{maierxxc}. In
another work the frequency dependence of effective properties have been
mathematically proven to deliver the leading-order dispersive behavior for
subwavelength plasmonic composites, this is rigorously  done in Theorem 4
of \cite{ChenLipton}. Last it is noted that the Lorentz resonance for a
single particle is not sufficient for understanding periodic subwavelength
patterned arrays of inhomogeneities as it ignores close range
inter-particle interactions that are captured by the local fields that
determine the effective dielectric constant.

\subsection{Summary of the main result}\label{sec:summary}

The objective of our discussion is to decouple the frequency dependence
introduced in \eqref{eq:intro:effective_perm} by the surface conductivity
and other material parameters from the geometric resonances of the
nanostructure. To this end we introduce an auxiliary spectral problem to
identify all $\big\{\lambda_n\big\}\subset \mathbb{C}$ for which there
exists a $\varphi_n$ satisfying
\begin{align*}
  \begin{cases}
    \begin{aligned}
    \Delta \varphi_n(\vx)\big)
      &= 0
      && \text{in }Y\setminus\Sigma~,
      \\[0.5em]
    \js{\varphi_n(\vx))} 
      &= 0
      && \text{on }\Sigma~,
      \\[0.5em]
      \lambda_n\js{\vn\cdot\nabla\varphi_n(\vx)}
      &=
      \nabla_T\cdot\nabla_T\varphi_n(\vx)
      && \text{on }\Sigma~.
    \end{aligned}
  \end{cases}\hspace{-2.75em}
\end{align*}
Introducing $\eta(\omega)=\frac{\sigma(\omega)}{\im\omega}$ we then show
that the effective refractive index in \eqref{eq:intro:effective_perm} can
be expressed by the formula
\begin{align}
  \label{eq:intro:characterization}
  \eff_{ij}(\omega)\;=\;
  \varepsilon\,\delta_{ij}\;-\;
  \eta(\omega)\,\int_\Si P_T(\vec e_j)\cdot P_T(\vec e_i)\dox
  \;-\;
  \sum_{n=1}^\infty\,
  \frac{\lambda_n\,\eta^2(\omega)}{\e-\lambda_n\,\eta(\omega)}
  \; M_{jn}
  \; \overline{M_{in}},
\end{align}
where the factors $M_{jn}$ are defined as
\begin{align*}
  M_{jn} = \int_\Si P_T(\vec e_j)\cdot\nabla_T\overline\varphi_n(\vx)\dox,
  \quad j=1,2,3,\quad n=1,2,\ldots
\end{align*}
The important property of this formula is that the integrals  only depend
on geometry, and the coefficients only depend on frequency. Equating the
real part of the denominator in the coefficients of
\eqref{eq:intro:characterization} to zero, recovers an explicit resonance
frequency $\omega_{R,n}$ for which the contribution of the $n$-th term of
the sum may become dominant,
\begin{align*}
  \omega_{R,n}\;=\;\sqrt{\omega^2_{0,n}-1/(2\tau)^2},
  \quad
  \text{where}\;\;
  \omega^2_{0,n}\;=\;\frac{\lambda_n\omega_p}{\e},
  \qquad n=1,2,\ldots
\end{align*}

\subsection{Past works}

 {%
Plasmonic crystals based on patterned dispersive dielectric 2D material
inclusions has made possible an unprecedented wealth of novel functional
optical devices~\cite{tan2020, mattheakis2016, silveirinha2007, moitra2013,
niu2018, xia2014}. Possible applications range from optical
holography~\cite{wintz2015}, tunable metamaterials~\cite{Nemilentsau2016},
and cloaking~\cite{alu2005}, to subwavelength focusing
lenses~\cite{cheng2014}.
}%

The  {analytical} approach taken here is motivated by earlier
observations of local resonances occurring at the length scale of the
microgeometry. Electrostatic resonances identified at the length scale of
composite geometry were shown to control the effective dielectric responce
associated with crystals made from non-dispersive dielectric inclusions in
the pioneering work of \cite{BergmanC} and \cite{MiltonES}. The associated
representation formulas based on local resonances were extended and applied
to bound the effective dielectric response \cite{MiltonCL},
\cite{GoldenPap}, and \cite{Milton}. Most recently local electrostatic and
plasmonic resonances are used to construct non-magnetic double negative
metamaterials in the near infrared \cite{ChenLipton} and design photonic
band gap materials \cite{RobertRobert1}.

The current work advances the understanding of effective dielectric
behavior by discovering and subsequently taking advantage of local
resonances supported both on surfaces and in the bulk for generating
Lorentz resonances at frequencies explicitly controlled by the
microstructure.

\subsection{Paper organization}

The remainder of the paper is organized as follows. In
Section~\ref{sec:spectral} we introduce the analytical setting and discuss
our spectral decomposition result. The emerging Lorentz resonance and an
application to inverse optical design in discussed in
Section~\ref{sec:physics}. A computational framework based on the spectral
decomposition is outlined in Section~\ref{sec:numerical} and two
prototypical geometries are  {computationally} analyzed. We
discuss implications and conclude in Section~\ref{sec:discussion}.

Analytical technicalities concerning the spectral decomposition result on
closed and open surfaces are outlined in Appendices~\ref{app:closed} and
\ref{app:open}. We summarize some explicit analytical formulas for the
solution of the geometric eigenvalue problem in
Appendix~\ref{app:examples}.


\section{Spectral Decomposition}
\label{sec:spectral}

In this section we introduce and characterize an auxiliary spectral problem
that enables us to derive the spectral decomposition
\eqref{eq:intro:characterization} of the cell problem
\eqref{eq:intro:strong_full}. For the sake of argument we keep the
discussion in this section on a formal level. A mathematically rigorous
formulation of the spectral decomposition for general classes of
\emph{closed} and \emph{open} 2D dielectric inclusions $\Si$ is given in
Appendix~\ref{app:closed} and \ref{app:open}, respectively. Here, a closed
inclusion $\Sigma$ is a $Y$-periodic two-dimensional surface that does not
have any one-dimensional edges in the interior of $Y$. Similarly, an open
inclusion $\Sigma$ is a $Y$-periodic two-dimensional surface that exhibit
an edge in the interior of $Y$; see Figure~\ref{fig:unit_cell}.

\subsection{An auxiliary eigenvalue problem}
As a first step we introduce an auxiliary eigenvalue problem that is
closely related to the cell problem \eqref{eq:intro:strong_full} of the
homogenization process. By removing the forcing $P_T\vec e_j$ and replacing
 {the quotient $\tfrac{\im\omega\e}{\sigma}$}  by a real-valued
eigenvalue $\lambda$ one arrives at the spectral problem: Find all pairs of
eigenvalues $\lambda\in\mathbb{R}$ and corresponding square-integrable
eigenfunctions $\varphi$ such that
\begin{align}
  \begin{cases}
    \label{eq:spectral_problem}
    \begin{aligned}
      \Delta \varphi(\vx)
      &= 0
      && \text{in }Y\setminus\Sigma~,
      \\[0.5em]
      \js{\phi(\vx)}
      &= 0
      && \text{on }\Sigma~,
      \\[0.5em]
      \lambda\js{\vn\cdot\nabla\varphi(\vx)}
      &=
      \Ds\varphi(\vx)
      && \text{on }\Sigma~.
    \end{aligned}
  \end{cases}\hspace{-2.75em}
\end{align}
Here, $\vn$ is again the unit outward normal of $\Si$ at $\vx$. We have set $\Ds:=\nabla_T\cdot\nabla_T$
and $\left[f\right](\vx)$ denotes the jump of a quantity $f$ across the
surface $\Si$ along the normal direction of $\Si$, viz.,
\begin{align*}
  \left[f\right](\vx)
  \,:=\,
  \lim_{\alpha\searrow0}
  \Big(f(\vx+\alpha\vn) - f(\vx-\alpha\vn)\Big)\quad
  \text{for }
  \vx\in\Si.
\end{align*}

Eigenvalue problem \eqref{eq:spectral_problem} is certainly well posed and
will admit an orthonormal basis of square-integrable eigenfunctions
provided one can identify an underlying self-adjoint and compact linear
operator. For all square-integrable \emph{densities} $\gamma(\vx)$ defined
on the surface $\Si$ we thus introduce the periodic single layer operator
$\bS\gamma$ by setting
\begin{align}
  \label{eq:bS}
  (\bS\gamma)(\vx) := \int_\Si \Gper(\vx-\vy) \gamma(\vy) \doy,
  \quad \vx\in Y.
\end{align}
Here, $\Gper$ is the periodic Green's function of the periodic Laplace
problem, viz.,
\begin{align*}
  \Gper(\vx) := \sum_{\vz\in\mathbb Z^n}G_0(\vz + \vx),
  \qquad G_0(x) := -\frac{1}{4\pi |x|}.
\end{align*}
The single layer operator $\bS$ is constructed in such a way that
$\bS\gamma$ satisfies
\begin{align}
  \label{eq:single-layer-potential}
  \Delta\bS\gamma=0\quad\text{in }\Omega\setminus\Si,
  \qquad
  \js{\bS\gamma}=0\quad\text{on }\Si,
  \qquad
  \js{\vn\cdot\nabla(\bS\gamma)}=\gamma\quad\text{on }\Si.
\end{align}
An important insight (that we outline in Appendix~\ref{app:closed}) is
the fact that this process can be reversed: In particular, for every
eigenfunction $\varphi$ that solves \eqref{eq:spectral_problem} one can
find a density $\gamma$ such that $\varphi(\vx)=(\bS\gamma)(\vx)$. This
allows us to substitute the representation $\varphi=\bS\gamma$ into the
last equation of \eqref{eq:spectral_problem}:
\begin{align*}
  \lambda\gamma(\vx) \;=\; \Ds (\bS\gamma)(\vx) \quad \text{on } \Sigma.
\end{align*}
Let $S$ denote the single layer operator $\bS$ restricted to $\Si$ and set
$\xi=S\gamma$. Provided that the inverses $S$ and $\Ds^{-1}$ exist we can
further rearrange the eigenvalue problem \eqref{eq:spectral_problem} into
an equivalent spectral problem:
\begin{align}
  \label{eq:eigenvalue_problem}
  \Ds^{-1}S^{-1}\xi \;=\; \lambda^{-1} \xi~.
\end{align}
We establish in Appendix \ref{app:closed} that for the case of closed
surfaces $\Sigma$ both inverses $S^{-1}$ and $\Ds^{-1}$ do indeed exist and
that the operator $\Ds^{-1}S^{-1}$ is compact and self adjoint on a
modified Hilbert space
\begin{align*}
  N(\Si)
  \;:=\; \left\{
    \xi\in H^{1}(\Si)
    \;:\;
    \int_{\Si} S^{-1}\xi\dox = 0
  \right\},
\end{align*}
with associated norm $\|\nabla_T\cdot\|_{L^2(\Sigma)}$. Here, $H^{1}(\Si)$
denotes the Sobolev space of square integrable functions with
square-integrable generalized derivatives. In summary, this guarantees the
existence of a countable set of real eigenvalues $\{\lambda_n^{-1}\}$,
$n=1\color{black}$, $2$, \ldots converging to zero, and an
associated orthonormal basis of eigenvectors $\{\xi_n\}$ of $N(\Si)$. Note
that by design $\xi_n$ is precisely the restriction of $\varphi_n$ as
characterized by \eqref{eq:spectral_problem} to the surface $\Si$.

\subsection{Spectral characterization of the corrector}
\label{subse:characterization}
Consider now the $Y$-periodic corrector field $\chi(\vx)$, described by the
cell problem \eqref{eq:intro:strong_full}. The aforementioned orthonormal
basis of eigenvectors $\{\xi_n\}$ admits (up to a constant) a
representation
\begin{align*}
  \chi_j(\vx) = \sum_{n=1}^\infty\alpha^n_j \xi_n(\vx)\quad\text{on }\Sigma.
\end{align*}
Substituting this characterization back into \eqref{eq:intro:strong_full}
and a bit of algebra exploiting \eqref{eq:eigenvalue_problem} then yields an
explicit formula for the coefficients:
\begin{align}
  \label{eq:coefficients}
  \alpha_j^n\;=\;
  \frac{\lambda_n\eta(\omega)}{\e-\lambda_n\,\eta(\omega)}
  \int_\Si P_T(\vec e_j)\cdot\nabla_T\overline\xi_n\dox,
  \qquad
  \eta(\omega)\;:=\;\frac{\sigma(\omega)}{\im\omega},
  \qquad
  j=1,2,3.
\end{align}
Similarly, repeating the substitution for Equation
\eqref{eq:intro:strong_full} yields an explicit formula for the frequency
behavior of the effective dielectric tensor:
\begin{align}
  \label{eq:effective_perm_new}
  \eff_{ij}(\omega)\;=\;
  \varepsilon\,\delta_{ij}\;-\;
  \eta(\omega)\,\int_\Si P_T(\vec e_j)\cdot P_T(\vec e_i)\dox
  \;-\;
  \sum_{n=1}^\infty\,
  \frac{\lambda_n\,\eta^2(\omega)}{\e-\lambda_n\,\eta(\omega)}
  \; M_{jn}
  \; \overline{M_{in}},
\end{align}
where the factors $M_{jn}$ are defined as
\begin{align*}
  M_{jn} = \int_\Si P_T(\vec e_j)\cdot\nabla_T\overline\varphi_n(\vx)\dox,
  \quad j=1,2,3,\quad n=1,2,\ldots
\end{align*}

A number of remarks are in order. Eq.~\eqref{eq:effective_perm_new}
separates the frequency dependence of the surface conductivity (included in
$\eta(\omega)$ {)} from the fundamental (frequency
independent) geometric resonances described by eigenvalue problem
\eqref{eq:spectral_problem} that determine eigenvalues $\lambda_k$ and
eigenmodes $\xi_k$. This implies that material properties and geometric
resonances, that both contribute to the frequency response of the effective
permittivity tensor $\eff_{ij}(\omega)$, are decoupled: The spectrum
$\lambda_k$ and eigenmodes $\xi_k$ of \eqref{eq:spectral_problem} only
depend on the (nanoscale) geometry and not on the concrete surface
conductivity model.


\section{Macroscale frequency response}
\label{sec:physics}

We now investigate the effective permittivity tensor $\eff(\omega)$ given
by \eqref{eq:effective_perm_new} further: First, the somewhat hidden
Lorentz resonance structure in the coefficients of the sum in
Eq.~\eqref{eq:effective_perm_new} is made explicit. We then discuss how
Eq.~\eqref{eq:effective_perm_new} can be used to facillitate the inverse
design process \cite{Miller2015, molesky}.

\subsection{ {Lorentz dispersive material}}

From \eqref{eq:effective_perm_new} we see that resonances in the temporal
behavior of the tensor emerge whenever the denominator in the coefficients
of the sum is close to zero. Equating the real part of the denominator of
\begin{align*}
  \frac{\lambda_n\eta^2(\omega)}{\varepsilon-\lambda_n\eta(\omega)}
\end{align*}
to zero recovers a critical frequency
\begin{align}
  \omega_{0,n}^2\;=\frac{\lambda_n\omega_p}{\varepsilon}.
\end{align}
Here, we assumed a simply Drude model
$\sigma(\omega)=\frac{\im\omega_p}{\omega+\im/\tau}$ to hold, where
$\omega_p\approx\frac{4}{137}$ is a rescaled Drude weight and $\tau$ is a
material-dependent relaxation time. With this definition in place we can
further manipulate the coefficients:
\begin{align*}
  \frac{\lambda_n\eta^2(\omega)}{\varepsilon-\lambda_n\eta(\omega)}
  \;=\;
  \frac{\omega_{0,n}^2\omega_p}{\omega(\omega+\im/\tau)}
  \;\cdot\;\frac{1}{\omega^2-\omega_{0,n}^2+\im\omega/\tau}\,.
\end{align*}
In general, the angular frequency $\omega$ is much larger than the inverse
of the relaxation time, viz., $\omega\gg1/\tau$. Thus, close to resonance
we can reasonably assume that
$\omega\approx\omega+\im/\tau\approx\omega_{0,n}$ and obtain
\begin{align*}
  \frac{\lambda_n\eta^2(\omega)}{\varepsilon-\lambda_n\eta(\omega)}
  \;\approx\;
  \omega_p \;\frac{1}{\omega^2-\omega_{0,n}^2+\im\omega/\tau}
  \;=\;
  -\omega_p\;\frac{\big(\omega_{0,n}^2-\omega^2\big)-\im\omega/\tau}
  {\big(\omega_{0,n}^2-\omega^2\big)^2+\omega^2/\tau^2}.
\end{align*}
The coefficients are thus Lorentzian with resonance frequency
\begin{align}
  \omega_{R,n}\;=\;\sqrt{\omega^2_{0,n}-1/(2\tau)^2},
  \qquad n=1,2,\ldots
\end{align}
In summary, we obtain that up to a constant:
\begin{align*}
  \eff_{ij}(\omega)\;\sim\;
  \; M_{jn}\overline{M_{in}}\,\omega_p
  \;\frac{\big(\omega_{0,n}^2-\omega^2\big)-\im\omega/\tau}
  {\big(\omega_{0,n}^2-\omega^2\big)^2+\omega^2/\tau^2},
  \quad\text{for}\;\omega\approx\omega_{R,n}.
\end{align*}
We point out that the nanoscale geometry that determines the spectrum
$\lambda_n$ and eigenmodes $\xi_n$ only influences the numerical values of
the resonance frequencies $\omega_{0,n}$ and $\omega_{R,n}$ as well as the
numerical values of the weights $M_{in}$.

 {%
  In summary, this heuristic argument suggests that the macroscale optical
  response of the plasmonic crystal is that of a \emph{Lorentz dispersive
  material} \cite{lorentz1916, KittelBook2005, born2013}: The frequency
  response of the effective permittivity, $\eff(\omega)$, can be
  approximated by a finite sum of Lorentz resonances, with explicit
  formulas for resonant frequencies and coefficients provided by our
  characterization \eqref{eq:effective_perm_new}.
}

\subsection{Inverse optical design}

The rational expression for the effective property given by
\eqref{eq:effective_perm_new}  lends itself to an inverse optimal design
paradigm for a desired dispersive response. Equation
\eqref{eq:effective_perm_new}  shows that the  location of the poles and
 {zeros} are controlled by the eigenvalues $\lambda_n$ and
eigenfunctions $\zeta_n$. These depend on the radius of the nano tubes or
the side-length of the nano ribbon. One can compute a library of
$\lambda_n$ and $\zeta_n$ near the desired operating frequency for a range
of geometric parameters. From this library one can pick the poles and
 {zeros} that delivers the dispersive response closest to the
desired one. Here the focus is on manufacturable designs and future
investigations will investigate libraries of manufacturable geometries to
provide a manufactorable range of resonant response.


\section{Computational platform}
\label{sec:numerical}

\begin{figure}[t]
  \centering
  \includegraphics{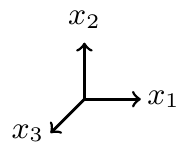}
  \subfloat[]{
%
    \includegraphics{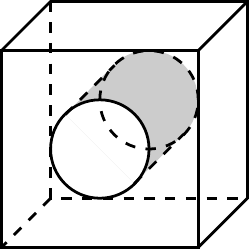}}
  \hspace{2em}
  \subfloat[]{
    \includegraphics{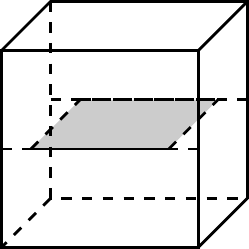}}
  \caption{Prototypical geometries: (a) shows a nanotube configuration and
    (b) is a nanoribbon configuration. The diameter (in (a)) and the width
    (in (b)) was set to 0.8.}
  \label{fig:geometries}
\end{figure}
The spectral decomposition discussed in
 {Section~\ref{sec:spectral}\ref{subse:characterization}}
allows for a very efficient computation of the frequency response of a
nanostructure by first solving a single geometric eigenvalue problem given
by \eqref{eq:spectral_problem} approximately. Then,
\eqref{eq:effective_perm_new} can be invoked to characterize the frequency
response of the permittivity tensor. We will illustrate this procedure in
this section on two prototypical geometries shown in
Figure~\ref{fig:geometries}: a nanotube configuration, which is a closed
smooth surface; and a nanoribbon configuration which is an open surface
with edges. We point out that due to the translation invariance in
$z$-direction of both configuration, the corresponding corrector
$\vec\chi_3$ vanishes. This implies that the corresponding cell problems
\eqref{eq:cell_problem} reduce to a two-dimensional problem, and that the
third diagonal component of the effective conductivity tensor $\eff$ is
simply given by
\begin{align*}
  \eff_{33}
  \;=\;
  \varepsilon\;-\;
  \eta(\omega)\,\int_\Si 1 \dox.
\end{align*}
Due to symmetry we have $\eff_{11}=\eff_{22}$ for the nanotube
configuration. In case of the nanoribbon geometry the averaging process in
$y$-direction is trivial leading to $\eff_{22}=\e$. We thus only need to
determine $\eff_{11}$ computationally in the following.


\subsection{Numerical computation of the geometric spectrum}
\label{subse:comparison}

\begin{table}[tb]
  \centering
  \subfloat[Nanotubes]{
    \footnotesize\begin{tabular}{ccc}
      \toprule
      order k & $\lambda_k$ & $\Big|\int_{\Si} P_T(\vec e_1)\cdot\nabla_T \overline\xi_n \dox\Big|$
      \\[0.5em]
      1 & 0.5924 & 1.1158      \\
      2 &  3.726 & 0.1077      \\
      3 &  6.289 & 0.008194    \\
      4 &  8.763 & 0.003574    \\
      5 &  11.26 & 0.0002755   \\
      6 &  13.76 & 0.00008546  \\
      7 &  16.27 & 0.000009443 \\
      \bottomrule
   \end{tabular}}
   \hspace{1em}
   \subfloat[Nanoribbons]{
    \footnotesize\begin{tabular}{ccc}
      \toprule
      $\lambda_k$ & $\Big|\int_{\Si} P_T(\vec e_1)\cdot\nabla_T \overline\xi_n \dox\Big|$
      \\[0.5em]
      0.9873 & 0.8543  \\
       5.314 & 0.1811  \\
       9.283 & 0.1097  \\
       13.22 & 0.07913 \\
       17.16 & 0.06194 \\
       25.02 & 0.04322 \\
       28.96 & 0.03755 \\
      \bottomrule
    \end{tabular}}
  \caption{
    Numerically computed spectrum and weight coefficients for the two
    geometries (Figure~\ref{fig:geometries}) using the computational
    approach outlined in Section~3(a): Table (a) shows results for the
    nanotube configuration. All roots have multiplicity 2; eigenvalues with
    weight 0 are omitted. Table (b) shows shows results for the nanoribbon
    geometry. Here all roots have multiplicity 1.}
  \label{tab:eigenvalues}
\end{table}

In order to approximate \eqref{eq:eigenvalue_problem} numerically we recast
the eigenvalue problem in the  {variational}
form~\eqref{eq:spectral_problem}: Find $\varphi_n\in\mathcal{H}$ and
$\lambda_n\in\R$ such that
\begin{align*}
  \lambda_n\,\int_Y\nabla\varphi_n(\vx)\cdot\nabla\overline\psi(\vx)\dx
  \;=\;
  \int_\Si\nabla_T\varphi_n(\vx)\cdot\nabla_T\overline\psi(\vx)\dox,
  \quad\forall\psi\in\mathcal{H}.
\end{align*}
This eigenvalue problem can be efficiently approximated with a finite
element discretization which we will quickly outline. We use the finite
element toolkit deal.II~\cite{dealIIcanonical,dealII92}. To achieve a good
numerical convergence order we use unstructured quadrilateral meshes
$\mathcal{T}_h$ for both geometries that are fitted to the curved
hypersurface $\Sigma$ by aligning element boundaries with the
hypersurface~\cite{Maier17} and discretize with high-order Lagrange
elements. Let $\big\{\psi^h_i\big\}_{i\in\{1:\mathcal{N}\}}$ be the nodal
basis of the Lagrange ansatz. We can then define the usual stiffness matrix
$M=(m_{ij})$
\begin{align*}
  m_{ij} \;=\;
  \sum_{Q\in\mathcal{T}_h}
  \int_Q\nabla\psi^h_j(\vx)\cdot\nabla\psi^h_i(\vx)\dx.
\end{align*}
The boundary term requires a modification because the trace
$\nabla_T\psi^h_i$ is not single-valued and only defined on an individual
cell of the mesh. We thus define a matrix $S=(s_{ij})$ by averaging both
cell contributions to the gradient:
\begin{align*}
  s_{ij} \;=\;
  \sum_{Q\in\mathcal{T}_h}
  \frac12
  \int_{\partial Q}\nabla_T\psi^h_j(\vx)\cdot\nabla_T\psi^h_i(\vx)\dox.
\end{align*}
We can then compute an approximate spectrum $\lambda_n^h$ and discrete
eigenfunctions $\xi^h_n = \sum_i\Xi_{n,i}^h\psi^h_i$ by solving the matrix
eigenvalue problem
\begin{align*}
  \big(S+bM\big)\,\Xi_n^h\;=\;\tilde\lambda_n^h\,M\,\Xi_n^h,\color{black}
\end{align*}
with an eigenvalue solver, such as SLEPc~\cite{Hernandez:2005:SSF}. Here,
$b>0$ is a suitably chosen Moebius parameter. The original eigenvalue is
recovered via by setting $\lambda_n^h=\tilde\lambda_n^h-b$. We briefly
comment on one crucial subtlety of this approach. The discrete eigenvectors
$\Xi^h_n$ are orthonormal with respect to the inner product $\langle
M\,.\,,\,.\,\rangle$ due to the mass matrix $M$ appearing on the right hand
side. This inner product is the discrete analogue of
$\int_Y\nabla\,.\cdot\nabla\,.\dx$ and not the normalization we used in
Section~\ref{sec:spectral}. This does not change the computed eigenvalues
but has an effect on the surface integrals that have to be computed next;
see Proposition~\ref{app:prop:spectral_decomposition} and the discussion in
Appendix~\ref{app:open}. This can be easily cured by scaling the surface
integrals appropriately by $1/\sqrt{\lambda^h_n}$, cf.
Equations~\eqref{eq:effective_perm_new} and
\eqref{app:eq:effective_perm_new}. We report numerical results for the two
geometries (Figure~\ref{fig:geometries}) in Table~\ref{tab:eigenvalues}.
The decay rate of the weight coefficients $\Big|\int_{\Si} P_T(\vec
e_1)\cdot\nabla_T \overline\xi_n \dox\Big|$ deserves a short discussion.
The rapid convergence of the coefficients to zero in case of nanotubes is
owed to the regularity of $\Si$ and the absence of interior edges. The
eigenvalues and eigenfunctions of the nanotube geometry can be explicitly
computed when the periodic boundary condition on $Y$ is replaced by an
infinite domain and the Sommerfeld radiation condition (see
Appendix~\ref{app:examples}). In this case only the first order, viz.
$k=1$, has a nonzero contribution to the resonance. The rapid decay of the
weight coefficients in our numerical result for the periodic case is
qualitatively in agreement with this observation. Due to the singularities
at the corners of the nanoribbon geometry~\cite{Rannacher1987}, it is not
surprising that the decay rate of the weight coefficients is limited.

An $n$-th order numerical approximation of the effective
permittivity tensor can be constructed by invoking a discrete counterpart
of \eqref{eq:effective_perm_new}:
\begin{multline}
  \label{eq:approximation}
  \e^{\text{app}}_{11}(\omega)\;=\;
  \varepsilon\,\;-\;
  \eta(\omega)\,
  \sum_{Q\in\mathcal{T}_h}
  \sum_{\partial Q\cap\Si} P_T(\vec e_1)\cdot P_T(\vec e_1)\dox
  \\[-1.0em]
  \;-\;
  \sum_{n=1}^N\,
  \frac{\lambda_n^h\eta^2(\omega)}{\e-\lambda_n^h\,\eta(\omega)}
  \; \left|
  \sum_{Q\in\mathcal{T}_h}
  \sum_{\partial Q\cap\Si} P_T(\vec e_1)\cdot\nabla_T \xi^h_n \dox\right|^2.
\end{multline}

\subsection{Comparison}
\label{subse:comparison}

\begin{figure}[htbp]
  \centering
  \subfloat[]{\includegraphics{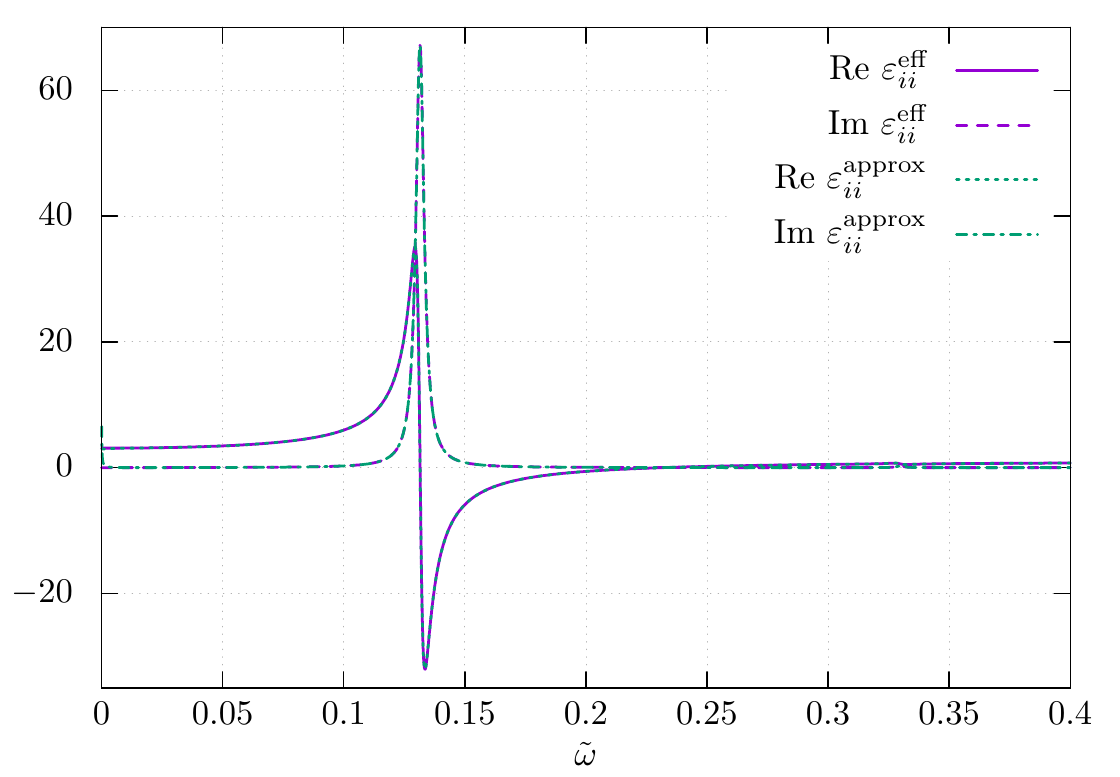}}
  
   \subfloat[]{\includegraphics{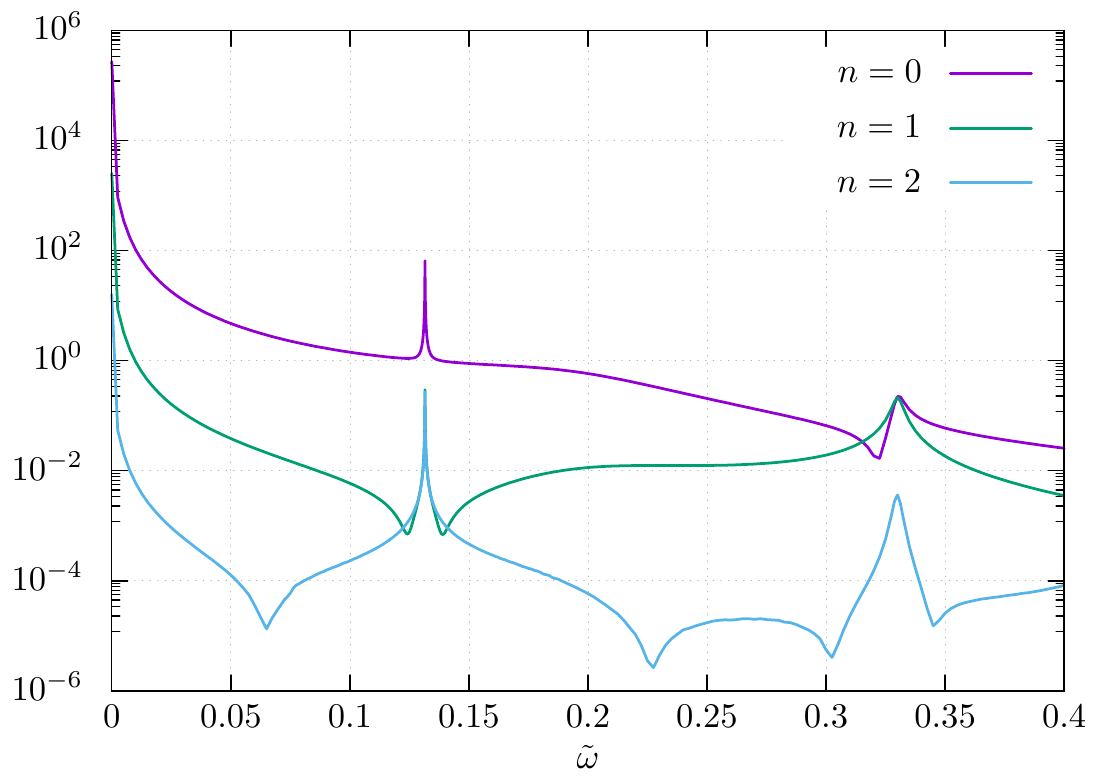}}
  \caption{%
    (a) Frequency response of $\eff_{ii}(\omega)$, $i=1,2$, for the
    nanotube configuration: The solid (real part) and dashed lines
    (imaginary part) are computed by solving the cell problem
    \eqref{eq:intro:strong_full} for every $\omega$; the dotted and
    dash-dotted lines are computed by formula \eqref{eq:effective_perm_new}
    truncated at $n=2$. (b) The corresponding relative error as a function of
    frequency.}
  \label{fig:tubes}
\end{figure}

\begin{figure}[htbp]
  \centering
  \subfloat[]{\includegraphics{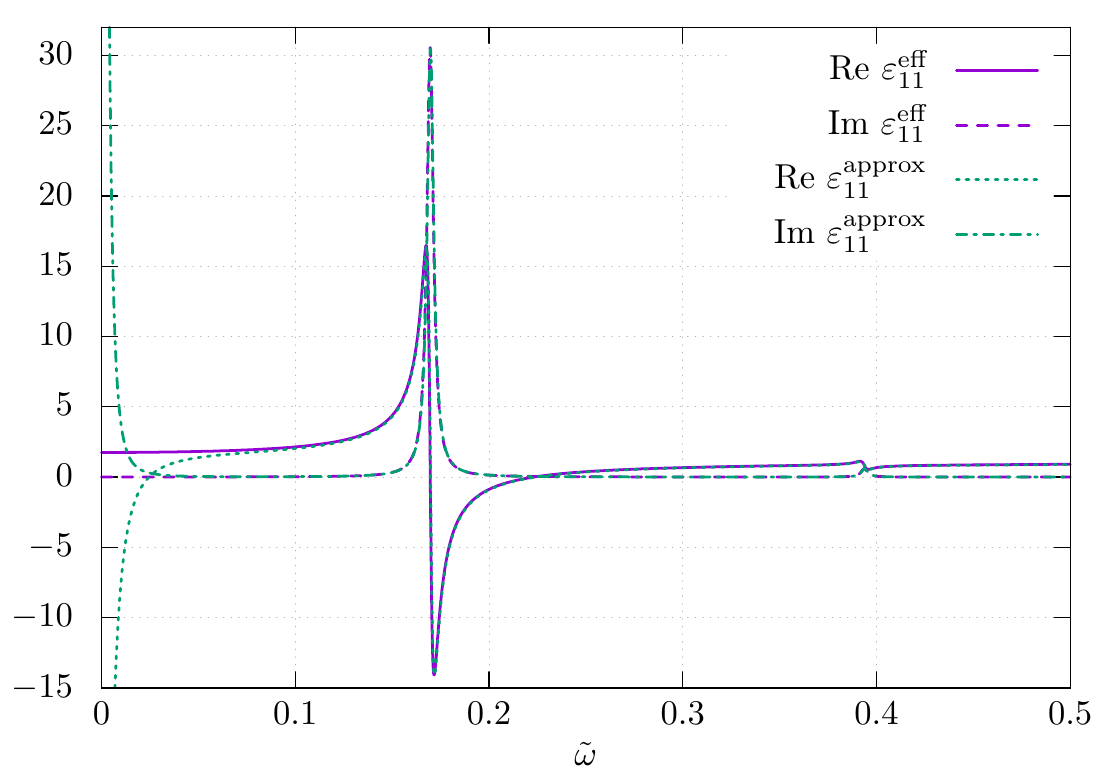}}

  \subfloat[]{\includegraphics{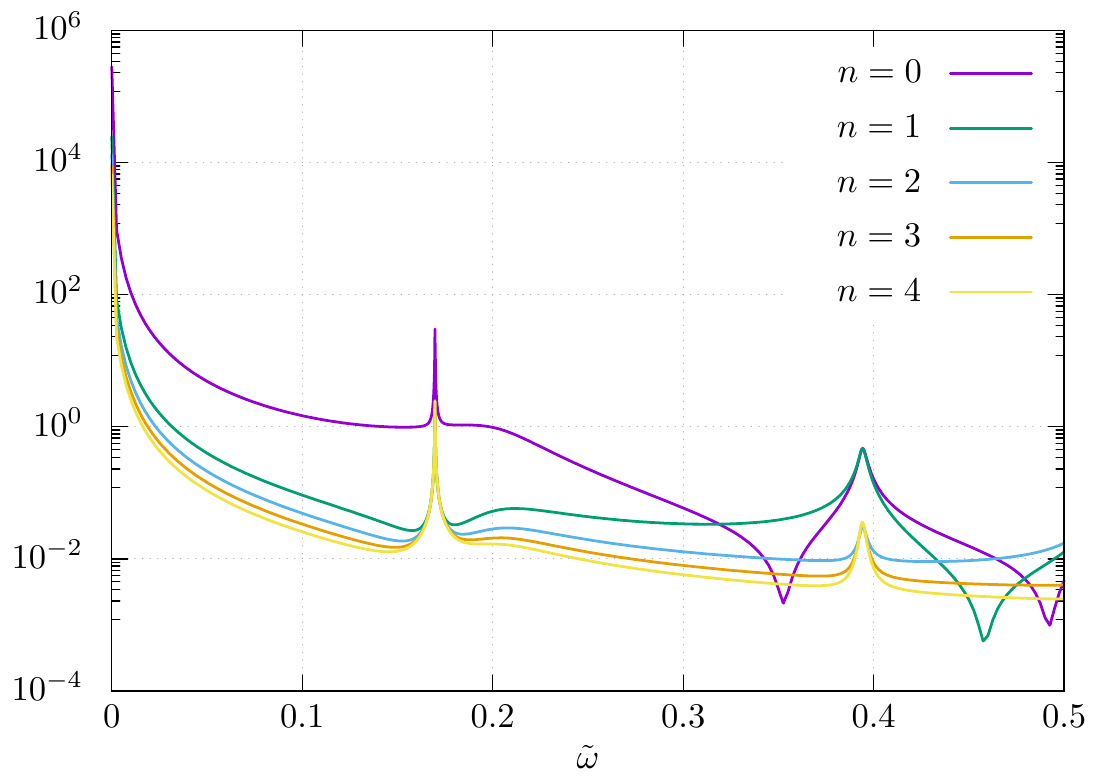}}
  \caption{%
    (a) Frequency response of $\eff_{11}(\omega)$, for the nanoribbon
    configuration: The solid (real part) and dashed lines (imaginary part)
    are computed by solving the cell problem \eqref{eq:intro:strong_full}
    for every $\omega$; the dotted and dash-dotted lines are computed by
    formula \eqref{eq:effective_perm_new} truncated at $n=2$. (b) The
    corresponding relative error as a function of frequency.}
  \label{fig:ribbons}
\end{figure}

Choosing $\e=1$, we compute a reference frequency response of $\eff_{11}(
\omega)$ by finely sampling over a set frequency range $0<\omega<0.5$ and
performing a complete direct numerical computation of the cell problem for
selected frequencies: For every chosen angular frequency $\omega$ we first
determine the corrector by solving \eqref{eq:intro:strong_full} with a
finite-element code \cite{maier19c} up to a suitable resolution (about
110\,k unknowns for the nanotube configuration, and about 130\,k unknowns
for the nanoribbon configuration). The result is plotted in
Figures~\ref{fig:tubes}a and \ref{fig:ribbons}a. In both plots about 700
frequencies were chosen adaptively.

We then compare a second order approximation $\e^{\text{app}}_{11}$ by
using \eqref{eq:approximation} with $n=2$ against the direct numerical
computation graphically in Figures~\ref{fig:tubes}a and \ref{fig:ribbons}a.
For the chosen frequency range we observe an excellent agreement of the
approximative permittivity $\e^{\text{app}}_{11}$ with the reference
computation in the ``eyeball'' norm.

A more detailed comparison of the frequency behavior of the relative error
between both computations is given in Figures~\ref{fig:tubes}b and
\ref{fig:ribbons}b, where also the dependence of the error on the order $n$
of the approximation \eqref{eq:approximation} is visualized. On average we
observe a relative error of less than $1\,\%$. We note that the maxima in
the relative error naturally occur at corresponding Lorentz resonances and
are dominated by the approximation error of the underlying finite element
simulations. We observe an exponential decay of the relative error as a
function of approximation order for the smooth nanotubes geometry (see
Figure~\ref{fig:tubes}b). The corresponding convergence behavior for
nanoribbons as shown in Figure~\ref{fig:ribbons}b is significantly slower.
This is owed to the fact that the edges in the nanoribbon geometry cause
singularity in the solution of the cell problems thus limiting the
approximation order \cite{Rannacher1987}.


\section{Conclusion}
\label{sec:discussion}

In this paper, we analyzed the Lorentz resonances in plasmonic
crystals that consist of 2D nano dielectric inclusions embedded in a
nonmagnetic bulk. From the corrector field found in a rigorous
homogenization theory \eqref{eq:cell_problem}, we derived an analytic
expansion formula for the effective permittivity
\eqref{eq:intro:characterization}. This formula decouples the local
geometric resonances and the material properties, and thus enables a very
efficient approximation to compute the frequency response. This formula
holds for inclusions of a large family of geometries, including closed
surfaces (as shown in Section~\ref{sec:spectral} and
Appendix~\ref{app:closed}) and open surfaces that can be completed into
closed surfaces {as shown in Appendix \ref{app:open}.

We observe that up to a constant factor, the $n$th Lorentz resonance is
described by
\begin{align*}
  \frac{\lambda_n\eta^2(\omega)}{\varepsilon-\lambda_n\eta(\omega)}
  \;\approx\;
  \omega_p \;\frac{1}{\omega^2-\omega_{0,n}^2+\im\omega/\tau}
  \;=\;
  -\omega_p\;\frac{\big(\omega_{0,n}^2-\omega^2\big)+\im\omega/\tau}
  {\big(\omega_{0,n}^2-\omega^2\big)^2+\omega^2/\tau^2}.
\end{align*}
We have also observed that a crucial quantity that determines the
convergence speed of this expansion is the decay rate of a weight factor
\begin{align*}
  \Big|M_{jn}\Big|\;=\;\Big|\int_{\Si} P_T(\vec e_j)\cdot\nabla_T \overline\xi_n \dox\Big|^2.
\end{align*}
The decay rate depends on the smoothness of the corrector, i.\,e., whether
singularities due to roughness or edges are present in the cell problem. We
have demonstrated that our spectral decomposition approach offers a
significant saving in computational resources because only one
frequency-independent geometric eigenvalue problem has to be computed, in
contrast to computing the corrector field for a huge number of fixed
frequencies.


\begin{appendices}

\section{Spectral Decomposition for closed surfaces}
\label{app:closed}

We now give a mathematical rigorous prove of the spectral decomposition
introduced in Section~\ref{sec:spectral} through the use of a weak
formulation. The mathematical proof involves simpler function spaces when
the surface $\Sigma$ is closed. For this reason we first discuss the case
of a closed surface and discuss the case of open surfaces based on the
notion of fractional Sobolev spaces in Appendix~\ref{app:open}.

\subsection{The weak formulation}
Provided the surface $\Sigma$ is Lipschitz  continuous (implying it admits
a uniquely defined surface normal), the $Y$-periodic corrector field
$\vec\chi(\vec x)$ is the solution of the variational \emph{cell problem}
\cite{maierxxa},
\begin{multline}
  \label{eq:cell_problem}
  \im\omega\varepsilon\,
  \int_Y\nabla\chi_j(\omega,\vx) \cdot\nabla\overline{\psi(\vx)} \dx
  \;-\;
  \sigma(\omega)\, \int_\Si \nabla_T\chi_j(\omega,\vx) \cdot
  \nabla_T\overline{\psi(\vx)} \dox
  \\
  \;=\;
  \sigma(\omega)
  \int_\Si
  P_T\big(\vec e_j) \cdot
  \nabla_T\overline{\psi(\vx)} \dox.
\end{multline}
The appropriate function space for the variational problem
\eqref{eq:cell_problem} is
\begin{align}
  \label{eq:correctorspace}
  \cH\;:=\;\Big\{\psi\in H^1_{\text{per}}(Y,\C)\,:\, \nabla_T\psi\in
  L^2(\Si,\C),\;\; \int_Y \psi = 0 \Big\}.
\end{align}
Here, $H^1_{\text{per}}(Y)$ denotes the Sobolev space of periodic functions
$u$ such that $u$ and its first order (distributional) partial derivatives
are square integrable in $Y$, and $L^2(\Si)$ denotes the space of
square-integrable functions on $\Si$. The space $\cH$ equipped with the
norm 
\begin{align*}
  \|\,\cdot \,\|_{\cH}^2= \|\nabla\,\cdot \,\|^2_Y + \|\nabla_T\, \cdot
  \,\|^2_\Si
\end{align*}
(and the corresponding inner product) is a Hilbert space. It can be shown
that the corrector problem \eqref{eq:cell_problem} admits a unique solution
$\chi_j\in\cH$ \cite{amirat2017, maierxxa, wellander2001, wellander2002,
wellander2003}.

Thus the auxiliary spectral problem partitioning between the first two
integrals in \eqref{eq:cell_problem}, parallel to
\eqref{eq:spectral_problem}, is to find all pairs of eigenfunctions
$\varphi\in\cH$ and eigenvalues $\lambda\in\mathbb{R}$, such that
\begin{align}\label{eq:spec_weak}
  \lambda\,\int_Y\nabla\varphi\cdot\nabla\overline\psi\dx
  \;=\;
  \int_{\Si}\nabla_T\varphi\cdot\nabla_T\overline\psi\dox
  \quad \text{for all } \psi\, {\in\cH}.
\end{align}

\subsection{A density representation for the corrector}
\label{sec:density}

The corrector $\chi_i\in\cH$ given by
\eqref{eq:cell_problem} can be characterized in terms of the $Y$-periodic
single layer potential $\bS$ \eqref{eq:bS} with a density $\gamma$.
Recall that we have restricted the discussion
to  the case of $\Si$ without internal edges in $Y$. In this case, the
following two properties hold:
\begin{enumerate}
  \item
    The restricted single layer operator $S\;:\;L^2(\Si)\to H^1(\Si)$
    defined by by \eqref{eq:single-layer-potential} is a bounded,
    invertible operator with a bounded inverse.
  \item
    The jump in the normal derivative of the solution $\chi_i\in\cH$ of
    the cell problem \eqref{eq:cell_problem} on the surface,
    $[\partial_{\vn} \chi_j]_{\Si}$, is in $L^2(\Si)$, where $L^2(\Si)$ the
    space of square integrable functions on $\Si$.
\end{enumerate}
A proof of (i) for the case of Lipschitz continuous $\Sigma$ can be found
in \cite[Thm.\,7.17]{McLean2000}, and Property (ii) is a direct consequence
of the standard trace theorems \cite{gilbarg2015elliptic} and Property (i).
Note that Properties (i) and (ii) do not hold when $\Si$ is an open surface
(i.e., when $\Sigma$ has edges in the interior of $Y$; see
Appendix~\ref{app:open}). Starting from (ii), we set
\begin{align*}
  \gamma\;:=\;[\partial_{\vn}\chi_j]\;\in L^2(\Sigma,\C).
\end{align*}
Recalling \eqref{eq:single-layer-potential} we observe that the difference
$\chi_j - \bS \gamma$ belongs to $H^{1}_{\text{per}}(Y,\C)$ and its
distributional Laplacian is zero everywhere in $Y$. Therefore,
\begin{align*}
 \chi_j = \bS \gamma+C,
\end{align*}
where $C$ is a constant. This suggests the following lemma:
\begin{lemma}
  \label{lemma:searchchi}
  For the corrector $\chi_j$ solving \eqref{eq:cell_problem}, there exists
  a unique $\gamma \in L^{2}(\Si,\C)$ and a unique complex valued constant
  $C$, such that
  \begin{align}
    \chi_j = \bS \gamma+C,
    \qquad\text{with }
    \int_{\Si}\gamma\dox\;=\;0.
  \end{align}
\end{lemma}

\begin{proof}
  We have already established existence. For the uniqueness, assume that we
  have two representations for $\chi_j$, viz. $\bS\gamma_1+C_1=
  \bS\gamma_2+C_2$. This implies that $\bS(\gamma_1 - \gamma_2)$ is a
  constant in $Y$, and thus
  \begin{align*}
    \gamma_1-\gamma_2 \;=\; [\partial_{\vn} \bS(\gamma_1-\gamma_2)] = 0
    \qquad\text{on }\Si.
  \end{align*}
  It follows that $C_1$ and $C_2$ are also identical.
  Finally, note that $\Delta\bS\gamma=0$ implies
  \begin{align}\label{eq:meanzero}
    \int_{\Si}\gamma \dox =
    \int_{\Si}[\partial_{\vn}\bS \gamma] \dox =
    -\int_{Y\setminus \Si}\Delta\bS\gamma dx = 0.
  \end{align}
\end{proof}

\subsection{An equivalent spectral problem and symmetrization}

Utilizing the same argument again as in the preceding subsection (Appendix
A\ref{sec:density}), for closed $\Sigma$ every eigenfunction $\varphi$ of
the spectral problem \eqref{eq:spec_weak} has a representation
$\varphi=\bS\gamma$, where $\gamma \in L^2(\Sigma)$. Substituting the
representation $\varphi=\bS\gamma$ into \eqref{eq:spectral_problem}, we
obtain an equivalent spectral problem for $\gamma$:
\begin{align*}
  \lambda\,\int_Y\nabla\bS\gamma\cdot\nabla\overline\psi\dx
  \;=\;
  \int_{\Si}\nabla_T\bS\gamma\cdot\nabla_T\overline\psi\dox , \quad
  \text{for all }
  \psi\, {\in\cH}.
\end{align*}
Integration by parts of the volume integral and \eqref{eq:single-layer-potential} further transforms the
eigenvalue problem to an eigenvalue problem described exclusively on $\Si$:
\begin{align*}
  -\lambda\,\int_{\Si}\gamma\,\overline{\psi}\dox
  \;=\;
  \int_{\Si}\nabla_T S\gamma\cdot\nabla_T\overline\psi\dox
  \quad \text{for all } \psi\,\in \cH.
\end{align*}
%
Writing
$\xi=S\gamma$, which is equivalent to $\gamma=S^{-1}\xi$ since $S\;:\;L^2(\Si)\to H^1(\Si)$ is invertible for closed $\Sigma$, we obtain
\begin{align*}
  -\lambda\,\int_{\Si} S^{-1}\xi\,\overline{\psi}\dox
  \;=\;
  \int_{\Si}\nabla_T \xi\cdot\nabla_T\overline\psi\dox , \quad \forall
  \psi\, {\in\cH}.
\end{align*}
%
%

\subsection{A compact and self-adjoint operator}

The property of the density function $\gamma$ in Lemma
\ref{lemma:searchchi} suggests that we work with the space
\begin{align}
  \label{eq:orthdecomp}
  N(\Si)
  \;:=\; \left\{
    \xi\in H^{1}(\Si)
    \;:\;
    \int_{\Si} S^{-1}\xi\dox = 0
  \right\}.
\end{align}
A straightforward calculation shows that $N(\Sigma)$ equipped with the norm
$\|\nabla_T \xi\|_{L^2(\Si)}$ is a Hilbert space. The Riesz representation
theorem then establishes a particular inverse of the Laplace-Beltrami
operator $\Delta_T$:
\begin{lemma}
  \label{lem:PDE2}
  For $f \in L^2(\Si)$ with $\int_{\Si}f\dox = 0$, there exists a unique
  $g\in N(\Si)$, such that
  \begin{align*}
   \int_{\Si}\nabla_T g \cdot\nabla_T\overline\psi\dox
    \;=\;
   - \int_{\Si}f \overline\psi\dox,\quad
    \text{for all }
    \psi\in  N(\Si).
  \end{align*}
  Moreover, the solution $g$ is bounded, viz., $\|\nabla_Tg\|_{L^2(\Si)}
  \leq C \|f\|_{L^{2}(\Si)}$. We will denote this solution operator by
  $\Ds^{-1}$.
\end{lemma}
We are now in a position to formulate and proof a central proposition and
corollary.
\begin{proposition}
  \label{prop:self-adjoint}
  The operator
  \begin{align*}
    \Ds^{-1} S^{-1}\;:\; N(\Si)\;\rightarrow\; N(\Si)
  \end{align*}
  is compact and self-adjoint. Moreover, $\text{ker}\,(\Ds^{-1} S^{-1}) =
  \left\{ 0 \right\}$.
\end{proposition}

\begin{corollary}[Spectrum]
  \label{cor:spectrum}
  The spectrum of $\Ds^{-1} S^{-1}$ consists of countably many nonzero
  eigenvalues $\big\{\lambda_n^{-1}\big\}_{n}$, only possibly accumulating
  at $0$. The corresponding eigenfunctions $\big\{\xi_n\big\}$ form an
  orthonormal basis of $N(\Si)$.
\end{corollary}

\begin{proof}[Proof of Proposition~\ref{prop:self-adjoint}]
  $\Delta_T^{-1}S^{-1}$ is well defined and bounded by virtue of Assumption
  (i) and Lemma~\ref{lem:PDE2}. For any given $f,g\in N(\Si)$, it holds that
  \begin{multline*}
    -\int_{\Si} \big(\nabla_T \Ds^{-1} S^{-1} f\big)\cdot\nabla_T g\dox
    \,=\,
    \int_{\Si} \big(S^{-1} f\big) g\dox
    \,=\,
    \int_{\Si} f \big(S^{-1} g\big)\dox
    \\
    \,=\,
    -\int_{\Si} \nabla_T  f\cdot\big(\nabla_T \Ds^{-1} S^{-1} g\big)\dox
  \end{multline*}
  Therefore, $\Ds^{-1} S^{-1}$ is self adjoint.

  In order to establish compactness of $\Delta_T^{-1}S^{-1}$, we first fix
  a bounded sequence $g_i\in N(\Si)$. The image $u_i=\Ds^{-1}S^{-1}g_i$ is
  also a bounded sequence in $N(\Si)$. By Rellich's lemma, there exists
  subsequence $u_{i_k}$ that is convergent in $L^2(\Si)$. Furthermore, we
  have
  \begin{align*}
    \int_{\Si}\nabla_T u_i \cdot \nabla_T u_j \dox
    \,=\,
    - \int_{\Si}S^{-1}g_i \,u_j \dox.
  \end{align*}
  Thus $\nabla_T u_{i_k}$ converges componentwise in $L^2(\Si)$, which
  gives that $u_{i_k}$ converges in $N(\Si)$.

  The last statement follows immediately from the fact that $S$ and $\Ds$
  are bounded and invertible. Thus $\Ds^{-1}S^{-1}f \equiv 0$ immediately
  implies $f\equiv S\,\Ds0 = 0$.
\end{proof}

\subsection{Proof of the spectral decomposition result}

\begin{proof}[Proof of \eqref{eq:effective_perm_new}]
  Let $\chi_j$ be the solution of \eqref{eq:cell_problem}. According to
  Lemma~\ref{lemma:searchchi} we can write $\chi_j$ as a single layer
  potential with a density $\gamma\in L^2(\Si)$ that satisfy
  $\int_{\Si}\gamma\dox =0$, viz.,
  \begin{align*}
    \chi_j \,=\, \bS\gamma \,+\, C.
  \end{align*}
  Using the invertibility of $S$ we obtain that for $\xi = S\gamma \in N(\Si)$,
  \begin{align*}
    \chi_j \,=\, \bS\,S^{-1}\xi \,+\, C.
  \end{align*}
  Corollary~\ref{cor:spectrum} guarantees the existence of the 
  expansion $\xi=\sum_k\alpha_j^k\xi_k$ with
  $\{\alpha_j^n\}_n\in\ell^2(\C)$, which yields (up to a constant):
  \begin{align*}
    \chi_j = \bS S^{-1}\Big(\sum_{k}\alpha^k_j \xi_k\Big).
  \end{align*}
  Identity \eqref{eq:coefficients} follows directly from substituting this
  expansion into \eqref{eq:cell_problem} and testing with $\psi=\bS
  S^{-1}\xi_k$:
  \begin{align*}
    \eta(\omega) \int_{\Si} P_T\big(\vec e_j) \cdot
    \nabla_T\overline{\xi_k} \dox
    \;
    &=\;\sum_{n}( \varepsilon \alpha_j^n / \lambda_n - \alpha_j^n
    \eta(\omega)) \int_\Sigma\nabla_T \xi_n \cdot\nabla_T\overline{\xi_k}
    \dox
    \\
    &=\;\sum_{n}( \varepsilon \alpha_j^n / \lambda_n - \alpha_j^n
    \eta(\omega))\,\delta_{kn}
    \\
    &=\;\varepsilon \alpha_j^k/\lambda_k - \alpha_j^k \eta(\omega).
  \end{align*}
  Finally, Identity \eqref{eq:effective_perm_new} follows from a similar
  substitution using Eqs.~\eqref{eq:coefficients} and
  \eqref{eq:intro:effective_perm}.
\end{proof}


\section{Spectral decomposition on open surfaces}
\label{app:open}

When $\Si$ is an open surface, in the sense that $\Si$ has edges in the
interior of $Y$, the property in Sec. \ref{sec:spectral}\ref{sec:density},
$S\;:\;L^2(\Si)\to H^1(\Si)$ is invertible, no longer holds. A
counter-example is the fact that in two dimensional space, the non-periodic
single layer potential maps $\frac{1}{\sqrt{a-x^2}}$ to a constant function
on the interval $[-a,a]$ \cite{YanSloan1988}. This means that we cannot
write $\chi_j = \bS \gamma +C$ for some $\gamma\in L^2(\Si)$. However, this
representation is valid for $\gamma$ defined in a proper \emph{fractional}
Sobolev space. Thus, modifying the argument to fractional Sobolev spaces
makes it possible to obtain the same expansion
\eqref{eq:effective_perm_new}.

In this appendix, we collect all necessary modifications to the argument
outlined in Appendix \ref{app:closed}, provided that the mild assumptions hold true that 
$\Si$ has smooth boundary and  $\Si$ can be completed into a closed smooth surface
$\Silarge$.

\subsection{Sobolev spaces on open surfaces}
\label{app:sobolev}

We give a definition of Sobolev spaces defined on open surfaces following
the notations in \cite{McLean2000}. First, on a closed $C^{k,1}$ surface
$\Silarge$ in $\R^n$, where $k\geq0$ and $n>0$ are integers,
$H^s(\Silarge)$ is defined through charts and the Fourier transform for
$s\in[-k-1,k+1]$ \cite[P.98]{McLean2000}.

Let $\Si$ be an open subset of $\Silarge$, and, for simplicity, assume the
boundary of $\Si$ is smooth. For every real number $s\in\R$, we define
\begin{align*}
  H^s(\Si) &:= \left\{f: \Si \rightarrow \C \,|\, f \text{ has an
  extension }\tilde f \in H^s(\Silarge) \right\},
  \\
  \tilde H^s(\Si) &:= \text{ closure of } C^\infty_0(\Si) \text{ in }
  H^s(\Silarge).
\end{align*}
It is shown in  \cite[Thm.\,3.14, Thm.\,3.29, Thm.\,3.30]{McLean2000} that
when $\Si$ is a Lipschitz subset of $\Silarge$, for all $s\in\R$,
\begin{align*}
  &(\tilde H^s(\Si))'  = H^{-s}(\Si),
  \\
  &( H^s(\Si))' = \tilde H^{-s}(\Si),\\
  &\tilde H^{s}(\Si) =  \left\{f\in H^{s}(\Silarge)\,|\, \text{supp} f\subset
  \overline{\Si} \right\},
\end{align*}
and for  an integer  $m\in[0, k+1]$,
\begin{align*}
  H^m(\Si) &= \left\{ f:\Si \rightarrow \C \,|\, f \text{
  and its weak tangential derivatives up to order } m \text{ are in }
L^2(\Si) \right\}.
\end{align*}
Note that the above defined $H^s(\Si)$ and $\tilde H^{-s}(\Si)$ for $s\geq
0 $ are the same as those defined in \cite{Wendland1979,Stephan1987}.


\subsection{Spectral decomposition}

We can now modify the argument in Appendix~\ref{app:closed} as follows.
Since $\chi_j$ belongs to $H^1(Y)$, its distributional Laplacian is $0$,
and $[\partial_n\chi_j]=0$ on $\Silarge\setminus\Si$, we obtain the
standard result that
\begin{lemma}
  For the corrector $\chi_j$ solving \eqref{eq:cell_problem}, there exists a
  unique $\gamma \in \tilde H^{-1/2}(\Si,\C)$ and a unique constant $C$, such
  that
  \vspace{-1em}
  \begin{align*}
    \chi_j = \bS \gamma+C.
  \end{align*}
  This $\gamma$ satisfies that $\int_{\Si}\gamma\dox=0$.
\end{lemma}
The mapping property of $S$ on $\tilde H^{-1/2}(\Si)$ is given by
\begin{lemma}[\cite{Wendland1979,Stephan1987}]
  \label{lem:inverseopen}
  The single layer operator $S: \tilde H^{-1/2}(\Si)
  \rightarrow H^{1/2}(\Si)$ is bijective.
\end{lemma}

The proper Hilbert space to consider becomes
\begin{align}\label{eq:orthdecompopen}
  \mathcal{N}(\Si)
  \;:=\; \left\{ f\in H^{1/2}(\Si),\;\langle S^{-1}f, 1\rangle_{\Si} = 0
  \right\},
\end{align}
equipped with the inner product $\langle -S^{-1}\xi,\,\eta
\rangle_{\Si}$. Here, $\langle\cdot,\,\cdot\rangle_{\Si}$ is the $L^2(\Si)$
pairing, and we'll refer to $\langle -S^{-1}\xi,\,\eta \rangle_{\Si}$ as
the $S^{-1}$ inner product.

On this space, we consider the following inverse of $\Ds$.
\begin{lemma}[A particular inverse of $\Ds$]
  \label{lem:PDE2a}
  For $f \in \tilde H^{-1}(\Si)$ with $\langle f,\,1\rangle_{\Si}=0$, there
  exists a unique $g\in H^1(\Si)$ with $\langle S^{-1}g,\,1 \rangle_{\Si} =
  0$, such that
  \begin{align}\label{eq:PDE2}
   - \langle f,\,\psi \rangle_{\Si}
    \;=\;
    \int_{\Si}\nabla_T g \cdot\nabla_T\overline\psi\dox , \quad
    \text{for all }
    \psi\in  H^1(\Si).
  \end{align}
  Moreover, the solution $g$ of \eqref{eq:PDE2} is bounded,
  $\|g\|_{H^1(\Si)} \leq C \|f\|_{H^{-1}(\Si)}$. We will denote this
  solution operator by $\Ds^{-1}$.
\end{lemma}

\begin{proof}
{%
  Given $f \in \tilde H^{-1}(\Si)$ with $\langle f ,1\rangle_{\Si} = 0$, it
  follows from standard elliptic equation theory that there exists a unique
  $\tilde g\in H^1(\Si)$ with $\langle \tilde g ,1\rangle_{\Si} = 0$, such
  that
  \begin{align*}
   - \int_{\Si}f \overline\psi\dox
    \;&=\;
    \int_{\Si}\nabla_T \tilde g \cdot\nabla_T\overline\psi\dox
    \quad \forall \psi\in H^{1}(\Si),
    \\[0.5em]
    \|\tilde g\|_{H^1(\Si)} \;&\leq\; C \|f\|_{\tilde H^{-1}(\Si)}.
  \end{align*}
  Now let $g_0 = S^{-1} 1 \in \tilde H^{-1/2}(\Si)$ and define the constant
  \begin{align*}
    C(\tilde g)\;:=\;
    \langle S^{-1} \tilde g,1\rangle_{\Si} / \langle S^{-1} 1, 1
    \rangle_{\Si} = \langle \tilde g,g_0\rangle_{\Si} / \langle  1, g_0
    \rangle_{\Si}.
  \end{align*}
  The function $g\,:=\,\tilde g - C(\tilde g)$ obviously solves
  \eqref{eq:PDE2} and by construction $\langle S^{-1} g ,1\rangle_{\Si} =
  0$.} The bound follows from
  \begin{align*}
    \|C(\tilde g)\|_{H^1} = \|C(\tilde g)\|_{L^2} \leq C | \langle S^{-1}
    \tilde g,1\rangle|
    \leq C \|S^{-1} \tilde g\|_{L^2} \leq C \|\tilde g\|_{H^1} \leq C
    \|f\|_{\tilde H^{-1}} .
  \end{align*}
\end{proof}

Since  $\Ds^{-1} S^{-1}$ maps $N(\Si) \subset H^{1/2}(\Si)$ into
$H^{1}(\Si)\subset\subset  H^{1/2}(\Si)$, we can verify:
\begin{proposition}
  The operator
  \begin{align*}
    \Ds^{-1} S^{-1}\;:\; \mathcal{N}(\Si) \;\rightarrow\; \mathcal{N}(\Si)
  \end{align*}
  is compact and self-adjoint with respect to the $S^{-1}$ pairing. Here
  $\Ds^{-1}$ is the particular operator defined in Lemma \ref{lem:PDE2a}.
  Moreover,
  \begin{align*}
    \text{ker}\,(\Ds^{-1} S^{-1}) = \left\{ 0 \right\}.
  \end{align*}
\end{proposition}

Finally, the main result reads:

\begin{proposition}[Spectral decomposition for open surfaces]
  \label{app:prop:spectral_decomposition}
  Let $\chi_j$ be the solution of the cell problem~\eqref{eq:cell_problem}.
  Let $\{\xi_n,\lambda_n^{-1}\}_n$ be the orthonormal eigen-system of the
  operator $\Ds S$ in the space $\mathcal N(\Si)$. Then
  \begin{align*}
    \chi_j = \bS S^{-1}\Big(\sum_{n}\alpha^n_j \xi_n\Big)+C,
  \end{align*}
  where $C$ is a constant and
  \begin{align*}
    \alpha_j^n\;=\;
    \frac{\eta(\omega)}{\e-\lambda_n\,\eta(\omega)}
    \int_\Si P_T(\vec e_j)\cdot\nabla_T\overline\xi_k\dox.
  \end{align*}
  Furthermore,
  \begin{multline}
    \label{app:eq:effective_perm_new}
    \eff_{ij}\;=\;
    \varepsilon\,\delta_{ij}\;-\;
    \eta(\omega)\,\int_{\Si} P_T(\vec e_j)\cdot P_T(\vec e_i)\dox
    \\
    \;-\;
    \sum_{n}\,
    \frac{\eta^2(\omega)}{\e-\lambda_n\,\eta(\omega)}
    \; \int_{\Si} P_T(\vec
    e_j)\cdot\nabla_T \overline\xi_n \dox
    \; \int_{\Si}  \nabla_T\xi_n \cdot P_T(\vec e_i)\dox.
  \end{multline}
\end{proposition}
Note that the $S^{-1}$ inner product gives a different normalization of
$\xi_n$ hence different $\alpha_j^n$ values. In terms of the scaled
function $\tilde\xi_k:=\frac{\xi_k}{\sqrt{|\lambda_k|}}$,
\eqref{eq:spectral_problem} is satisfied and the expansion
\eqref{app:eq:effective_perm_new} takes the same form as
\eqref{eq:intro:characterization}.


\section{Explicitly computable examples}\label{app:examples}

We explicitly compute the eigensystem of $\Ds^{-1}S^{-1}$ on two {\em
nonperiodic} geometries in $\R^3$. These examples qualitatively illustrate
the corresponding periodic geometries, when the inclusions are far apart
from each other. On spheres and circular cylinders in $\R^3$, the
eigensystem of $\Ds S$ are explicitly known. This is because $\Ds$ and $S$
separately have explicit eigensystems, and they share eigenfunctions. Note
that the only manifolds on which the Laplace-Beltrami operator has explicit
eigensystems are $n$-spheres, $n$-tori and Heisenberg groups. 

\subsection{Circular cylinder}

Let $\Si$ be a cylinder with a circular cross section of radius $a$. The
corresponding periodic geometry is the nanotube structure considered
numerically in Sec. \ref{sec:numerical}. We will abuse notation by denoting
the cross sections of all quantities by the same notation, since all
quantities are invariant along the axis of the cylinder. A basis for mean
zero $L^2(\Sigma)$ functions is $\left\{e^{in\theta}, n\neq0\right\}$. This
is also a set of simultaneous eigen functions for $\Ds$ and $S$:
\begin{equation*}
  \Ds e^{in\theta}=-\frac{n^2}{a^2}e^{in\theta},\quad
  Se^{in\theta}=-\frac{a}{2n}e^{in\theta}.
\end{equation*}
Thus the eigensystem for \eqref{eq:eigenvalue_problem} normalized in the
$\|\nabla_T\cdot\|_{L^2(\Si)}$ norm is
\begin{equation*}
  \lambda_n = \frac{n}{2a},\quad \xi_n^{k\color{black}}=
  \begin{cases}
    \frac{1}{n}\sqrt{\frac{a}{\pi}}\cos(n\theta),\quad
    k\color{black}=1
    \\
    \frac{1}{n}\sqrt{\frac{a}{\pi}} \sin(n\theta),\quad
    k\color{black}=2
  \end{cases}
  ,\quad n\geq1.
\end{equation*}
Using $ P_T(\vec e_1) = - \hat \theta \sin\theta $ and $ \nabla_T =
\hat\theta \frac{1}{a}\partial_\theta$, we obtain
\begin{align*}
  \int_\Si P_T(\vec
  e_1)\cdot\nabla_T\overline\xi_n^{k\color{black}}\dox
  =
  \begin{cases}
  \sqrt{\pi a},  \quad n=1, \, k=1,
  \\
  0,  \quad \text{otherwise}.
  \end{cases}
\end{align*}
Note that the for the corresponding periodic geometry, the factor $
\int_\Si P_T(\vec
e_1)\cdot\nabla_T\overline\xi_n^{k\color{black}}\dox $ in
Table~\ref{tab:eigenvalues} decays, instead of falling to zero abruptly.
This is due to the effect from other cylinders in the array. The decay
becomes faster when the size of the cylinder relative to the cell becomes
smaller.

\subsection{Sphere}

Let $\Si$ be a sphere of radius $a$. A basis for mean zero $L^2(\Sigma)$
functions is the set of spherical harmonic functions $\left\{Y_n^m, n\geq1,
-n\leq m\leq n\right\}$. This is also a set of simultaneous eigenfunctions
for $\Ds$ and $S$:
\begin{equation*}
  \Ds Y_n^m=-\frac{n(n+1)}{a^2}Y_n^m,\quad SY_{n,m}=-\frac{a}{2n+1}Y_{n,m}.
\end{equation*} 
Thus the eigensystem for \eqref{eq:eigenvalue_problem} is normalized in the
$\|\nabla_T\cdot\|_{L^2(\Si)}$ norm is
\begin{equation*}
  \lambda_n = \frac{n(n+1)}{a(2n+1)},\quad \xi_n^i= \frac{1}{\sqrt{n(n+1)}}
  Y_{n,m}, \quad n\geq1,-n\leq m\leq n.
\end{equation*}
Using $ P_T(\vec e_1) = \hat \theta \cos\theta\cos\phi -\hat\phi\sin\phi $,
$ \nabla_T = \hat\theta\frac{1}{a}\partial_{\theta}
+\hat\phi\frac{1}{a\sin\theta}\partial_{\phi}$ and the recurrence relations
for the associated Legendre polynomials, we obtain
\begin{align*}
  \int_\Si P_T(\vec e_1)\cdot\nabla_T\overline\xi_n^i\dox
  =
  \begin{cases}
\mp 2a\sqrt{\frac{\pi}{3}},  \quad n=1, \, m=\pm 1,
  \\
  0,  \quad \text{otherwise}.
  \end{cases}
\end{align*}
\end{appendices}


\section*{Acknowledgments}

RL acknowledges partial support by the NSF under grant DMS-1921707 and 1813698, MM acknowledges partial support by the NSF under grant DMS-1912847.



\end{document}